\documentclass[sn-mathphys-num]{sn-jnl}% Math and Physical Sciences Numbered Reference Style 
\usepackage{graphicx}%
\usepackage{multirow}%
\usepackage{amsmath,amssymb,amsfonts}%
\usepackage{amsthm}%
\usepackage{mathrsfs}%
\usepackage[title]{appendix}%
\usepackage{xcolor}%
\usepackage{textcomp}%
\usepackage{manyfoot}%
\usepackage{booktabs}%
\usepackage{algorithm}%
\usepackage{algorithmicx}%
\usepackage{algpseudocode}%
\usepackage{listings}%

\usepackage{subcaption}

\newtheorem{theorem}{Theorem}%

\begin{document}

\title[The EC-SBM Synthetic Network Generator]{EC-SBM Synthetic Network Generator}

\author[1]{\fnm{The-Anh} \sur{Vu-Le}}\email{vltanh@illinois.edu}

\author[1]{\fnm{Lahari} \sur{Anne}}\email{lanne2@illinois.edu}
\author[1]{\fnm{George} \sur{Chacko}}\email{chackoge@illinois.edu}

\author*[1]{\fnm{Tandy} \sur{Warnow}}\email{warnow@illinois.edu}

\affil*[1]{\orgdiv{Siebel School of Computing and Data Science}, \orgname{University of Illinois Urbana-Champaign}, \orgaddress{\street{201 N. Goodwin Avenue}, \city{Urbana}, \postcode{61801}, \state{IL}, \country{USA}}}

 \abstract{
    Generating high-quality synthetic networks with realistic community structure is vital to effectively evaluate community detection algorithms. In this study, we propose a new synthetic network generator called the Edge-Connected Stochastic Block Model (EC-SBM). The goal of EC-SBM is to take a given clustered real-world network and produce a synthetic network that resembles the clustered real-world network with respect to both network and community-specific criteria. In particular, we focus on simulating the internal edge connectivity of the clusters in the reference clustered network. Our extensive performance study on large real-world networks shows that EC-SBM has high accuracy in both network and community-specific criteria, and is generally more accurate than current alternative approaches for this problem. Furthermore, EC-SBM is fast enough to scale to real-world networks with millions of nodes.
 }


\maketitle

\section{Introduction}
\label{sec:intro}

Models for synthetic network generation are in wide use, with applications that vary from proposing null models \cite{erdos-renyi, chung-lu, modularity, config-model} to evaluating graph analysis techniques such as community detection \cite{lfr, danon2005, orman2011, orman2012, orman2013, hric2014, yang2016, abcd, abcd+o} or community structure testing \cite{lancichinetti2010, bickel2015, miasnikof2020, miasnikof2023, yanchenko-sengupta2024, zengyou2025}.

Several of the early approaches do not incorporate community structure into the generated network. 
The Erdős–Rényi model \cite{erdos-renyi} connects nodes in an independently and identically distributed manner. The Barabási–Albert model \cite{barabasi-albert} captures the power-law degree distribution often observed in real-world networks. The Chung-Lu model \cite{chung-lu} and the configuration model \cite{config-model-first, config-model} enforce specific degree distribution, either in expectation or in reality. These models are usually simple and only focus on a specific property; however, they provide the foundation for much future research. 

Later models have incorporated community structure, including Stochastic Block Models (SBMs) \cite{sbm_first,sbm}, the Lancichinetti–Fortunato–Radicchi (LFR) benchmark \cite{lfr},  Artificial Benchmark for Community Detection (ABCD) \cite{abcd},  and the extension of ABCD called ABCD+o  to handle outliers \cite{abcd+o}, i.e., nodes that are not in any cluster.
The nonuniform Popularity-Similarity Optimization (nPSO) \cite{npso} model incorporates parameters not directly estimated from the reference graph but indirectly influencing other network and cluster properties of interest. 
Many of these approaches can take as input various parameters for a clustered network, such as the degree sequence, the number of edges and vertices in each cluster, the mixing parameter, etc. 
Furthermore, SBMs can take as input the assignment of nodes to clusters. 

More recently, over-parameterized models have emerged, such as the exponential family random graph models (ERGM) \cite{ergm, ergm-digraph}, GraphRNN \cite{graphrnn}, and other deep graph generative models \cite{graphgen_survey, graphgen_survey2}.
These models utilize parameters that are not directly tied to graph properties, like degree or edge counts, but serve as coefficients in complex models. Moreover, these parameters can be ``learned'' from a set of reference networks.  
While these approaches can capture more intricate network properties, including communities, they face challenges in scalability and scarcity of large-network data, limiting their applicability to smaller graphs.
For example, the study done by \cite{geel} showed GraphRNN and other deep generative techniques do not scale to networks with $5000$ nodes where they ran into out-of-memory errors. 

Here we ask how well a given network generator is able to produce synthetic networks that resemble a given real-world network.
This question was partially addressed by Vaca-Ramirez and Peixoto
\cite{eval_sbm}, who examined the \texttt{graph-tool} software \cite{graph-tool} for generating SBMs with respect to its accuracy in reproducing network features.  . 
Vaca-Ramirez and Peixoto \cite{eval_sbm} took a large number of real-world networks, computed a clustering of the network by finding a best-fitting SBM for the network, and then used that SBM to generate synthetic networks. 
That study then established that their SBM software produced synthetic networks that came closer to the real-world network in terms of various properties, such as degree sequence, diameter, local and global clustering coefficients, etc., than a configuration model.
Despite that positive result, \cite{eval_sbm} did not examine whether the SBMs produced by \texttt{graph-tool} had ground-truth communities that were similar to the communities in the clustered real-world network.

Very recently, Anne {\em et al.} \cite{reccs} examined the question of how well the clusters produced by \texttt{graph-tool} resembled the clusters in the input clustered network.
Their study, which examined a large number of networks clustered using the Leiden software \cite{leiden-code} optimizing either modularity or the Constant Potts model \cite{leiden}, established that the SBMs produced by \texttt{graph-tool} often had disconnected ground truth clusters.
Anne {\em et al.} also presented a new simulator to address this problem, REalistic Cluster Connectivity Simulator (RECCS).
Given the parameters from a clustered real-world network, RECCS creates an initial SBM using \texttt{graph-tool} and then modifies it to come close to the parameters in the input.
As shown in \cite{reccs}, RECCS produces synthetic networks and ground truth clusterings that have network and cluster properties similar to that of the input clustered real-world network and is especially good with respect to edge-connectivity (i.e., the size of a minimum edge cut) for each cluster.  Furthermore, in comparison to SBMs computed using \texttt{graph-tool}, RECCS improves the accuracy for edge-connectivity values and matches the network parameter values.

Here, we present a new method, Edge-Connected SBM (EC-SBM), for generating synthetic networks that aims to achieve even higher fidelity than RECCS. 
The design of EC-SBM has a different structure than RECCS in two major ways: while RECCS begins by creating an SBM synthetic network and then modifying it, EC-SBM begins by establishing some of the edges for the clusters that suffice to guarantee achieving edge connectivity that is at least that in the input parameters. 
Then EC-SBM uses \texttt{graph-tool} SBM construction to augment the edge set and afterward further modifies the network to improve the fit to the input values.
Our study, using a large corpus of real-world networks, also examines the question of which clusterings of real-world networks produce the closest fit to the real-world network.
We find distinct differences in final synthetic network quality (as measured using the fit between the synthetic network and the given real-world network) between clustering methods, with two specific clustering methods providing the best fit.
Thus, in this paper, we provide some advances in the development of protocols and software to produce realistic synthetic networks that resemble given real-world networks.

The rest of the paper is organized as follows.
In Section \ref{sec:prelim-study}, we begin with our preliminary study of synthetic networks using SBMs, and provide extensive evidence that SBMs produce disconnected ground truth clusters for many different ways of clustering real-world networks, thus expanding on the results provided in \cite{reccs}.
We also examine the frequency with which SBMs produce parallel edges and self-loops, which enables them to reproduce with high accuracy the degree sequence of the input network. 
We continue in Section \ref{sec:design-ecsbm} with the description of the EC-SBM network simulator, which has three stages.
In Section \ref{sec:results}, we present the experimental study and results, comparing SBM, EC-SBM, and RECCS on a large corpus of clustered real-world networks. In this section, we also evaluate different clustering methods to determine which ones are most suitable for use when seeking to reproduce the features of a clustered real-world network.
Finally, we close in Section \ref{sec:conclusions} with a discussion of future work.
Details on materials and methods are provided in Section \ref{sec:methods-materials}.

\section{Preliminary study}
\label{sec:prelim-study}

In this section, we explore how well SBMs are able to reproduce the features of a real-world network and, in particular, how often the generated ground truth clusters are disconnected.
We also examine the incidence of parallel edges and self-loops.
We begin with a description of the stochastic block model and its formulation within \texttt{graph-tool}.

\subsection{Stochastic Block Model (SBM)}
\label{sec:prior-works/sbm}

The \textbf{stochastic block model} (SBM) \cite{sbm} is a generative model for graphs. It can be used for both community detection and as a synthetic network generator; in this section, we will focus on the latter. Specifically, we base our analysis on the \texttt{generate\_sbm} function of \texttt{graph-tool} \cite{graph-tool}. We set \texttt{micro\_ers} and \texttt{micro\_degs} to \texttt{true} to evoke the micro-canonical version of the degree-corrected model. We set \texttt{directed} to \texttt{false} to generate only undirected graphs. For the rest of the text, we refer to this implementation as SBM.

SBM takes, as input parameters, 
\begin{itemize}
    \item the desired cluster (block) assignment as a vector $b$ of size $n \in \mathbb{N}$ (number of vertices) where each element is an integer from $1$ to $m \in \mathbb{N}$ (number of clusters) indicating to which cluster each vertex belongs
    \item the desired degree sequence as a vector $d \in \mathbb{N}^n$ indicating the degree of each vertex
    \item the desired edge count matrix as a matrix $e$ of size $m \times m$ indicating the number of edges between each pair of clusters with the off-diagonal values and double the number of edges inside each cluster with the values on the main diagonal.
\end{itemize}
SBM will produce a synthetic network with exactly $n$ vertices, all assigned to $m$ clusters exactly according to $b$. Moreover, the degree sequence and edge count matrix of the synthetic network and cluster assignment will be close to the desired (except for some instances discussed in the observation below).

Given a clustered real-world network, we use SBM as a simulator in an obvious way.
A given clustering will produce some clusters with the size of at least two and then some nodes that are not clustered with any other node.
These nodes also can be considered members of clusters, but where the clusters have only one member.
Given this, we can use SBM as a simulator to produce a synthetic network by passing the necessary parameters to SBM using this clustering that includes singleton clusters.
The set of parameters includes the degree sequence, the assignment of nodes to clusters, and the edge count matrix (which specifies the number of edges within each cluster and between every pair of clusters).

Note that this approach ensures that the synthetic cluster assignment exactly matches the desired cluster assignment: there is a one-to-one correspondence between the synthetic and empirical clusters and between the vertices of the synthetic and empirical networks. 

\subsection{Evaluating SBM synthetic networks}

We cluster each network in a corpus of small and medium-sized $74$ real-world networks (see Section \ref{sec:methods-materials/materials}) using different clustering methods to obtain input clusterings.
Some of these clusterings are post-processed using methods such as the Connectivity Modifier (CM) \cite{cm} and Well-Connected Clusters (WCC) \cite{sbm-wcc}, which split clusters into smaller clusters so as to improve their edge-connectivity.
Then, for each network-clustering pair, we compute the parameters and use SBM to generate a synthetic network. Finally, we compute the proportion of disconnected clusters and excess edges (i.e., self-loops and parallel edges) in these synthetic networks. Figure \ref{fig:sbm/disconnect_excess} summarizes the results, where the network-clustering pairs are grouped by the clustering method to produce the clustering.

\begin{figure}[!t]
    \centering
    \includegraphics[width=\linewidth]{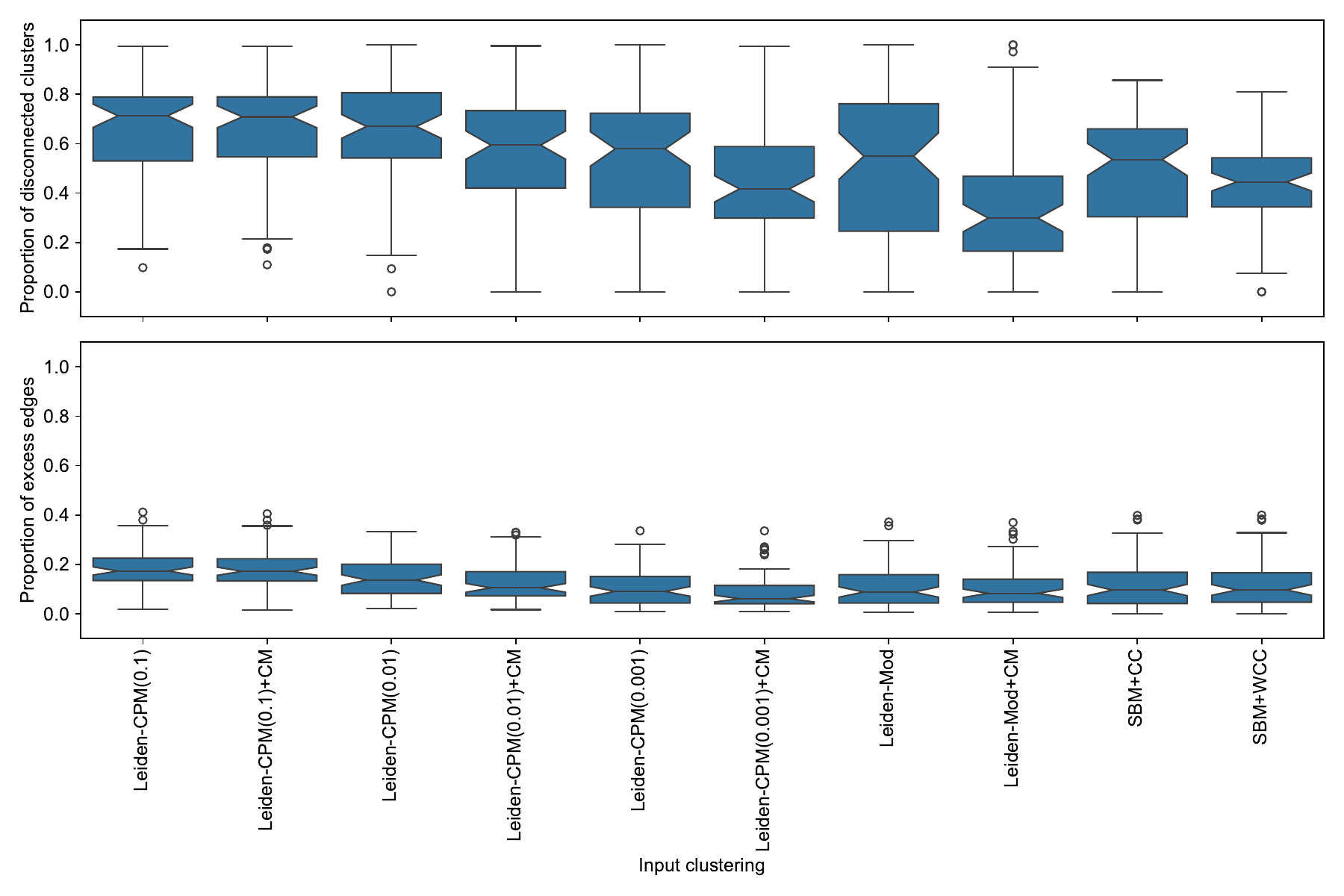}
    \caption{\textbf{Proportion of disconnected clusters (top) and excess edges (bottom) in synthetic networks generated by SBM.} SBM is given network and clustering statistics for $74$ networks, each clustered by one of the clustering methods specified on the horizontal axis. These clusterings are all guaranteed to produce connected clusters. SBM then uses these parameters to produce a synthetic network.
    This figure shows that while the choice of clustering method affects the proportion of clusters that are internally disconnected and the overall proportion of excess edges, SBM produces many disconnected clusters and many excess edges for all tested clusterings. 
    }
    \label{fig:sbm/disconnect_excess}
\end{figure}

From Figure \ref{fig:sbm/disconnect_excess} (top), we first observe that, while the choice of clustering method affects the proportion of clusters that are internally disconnected, the median proportion lies from around $30\%$ to around $70\%$. This indicates that SBM produces many disconnected clusters for all studied input clusterings, even though these clusterings are guaranteed to produce connected clusters. This observation motivates the focus of our approach to improve the preservation of the edge connectivity of the clusters, which indirectly prevents the generated clusters from being disconnected.

We also observe from Figure \ref{fig:sbm/disconnect_excess} (bottom) that the median proportion of excess edges lies from around $10\%$ to around $20\%$. This indicates that SBM produces a non-trivial number of excess edges. Removing the excess edges prevents SBM from closely matching the desired degree sequence. Since our method also uses SBM in intermediate steps, this observation motivates us to include a final step of degree correction to mitigate this problem.

\section{The EC-SBM Network Simulator}
\label{sec:design-ecsbm} 

In this section, we describe our proposed generator. We give a high-level outline of the process here and provide the details in Section \ref{sec:methods-materials}.

The input is a clustered real-world network $N$ and we want to produce a synthetic network and a synthetic cluster assignment that reflects that input. 
We begin by separating the nodes into two sets: those that are in non-singleton clusters and the remaining nodes, which we refer to as the ``outliers".
This division of the nodes into two sets defines two subnetworks.
The first is
the subnetwork of $N$ induced by the nodes in the first set, and so containing only those edges that connect pairs of such nodes, is called the clustered subnetwork for the input real-world network. 
The second is the subnetwork of $N$ that contains all the outlier nodes (i.e., the nodes in the second set) and their neighbors;  the edges in this subnetwork are those that have at least one endpoint in the second set (i.e., an outlier node).
Note that every edge in $N$ is in exactly one of these two subnetworks.

We split the generation process into three stages. In the first stage, we generate the synthetic clustered subnetwork, accompanied by a synthetic cluster assignment. In the second stage, we generate the outlier subnetwork. The union of the two subnetworks will be our partially complete synthetic network. In the third stage, we add edges to match the input degree sequence. 

\paragraph{Stage 1: Generation of the synthetic clustered subnetwork}

We maintain the same cluster assignment as the empirical network in the synthetic network. Then, we compute the minimum cut size of each empirical cluster, which indicates the desired edge connectivity of each synthetic cluster. The process of the entire stage is visualized in Figure \ref{fig:chambana-stage1}.

\begin{figure}[!t]
    \centering
    \includegraphics[width=\linewidth]{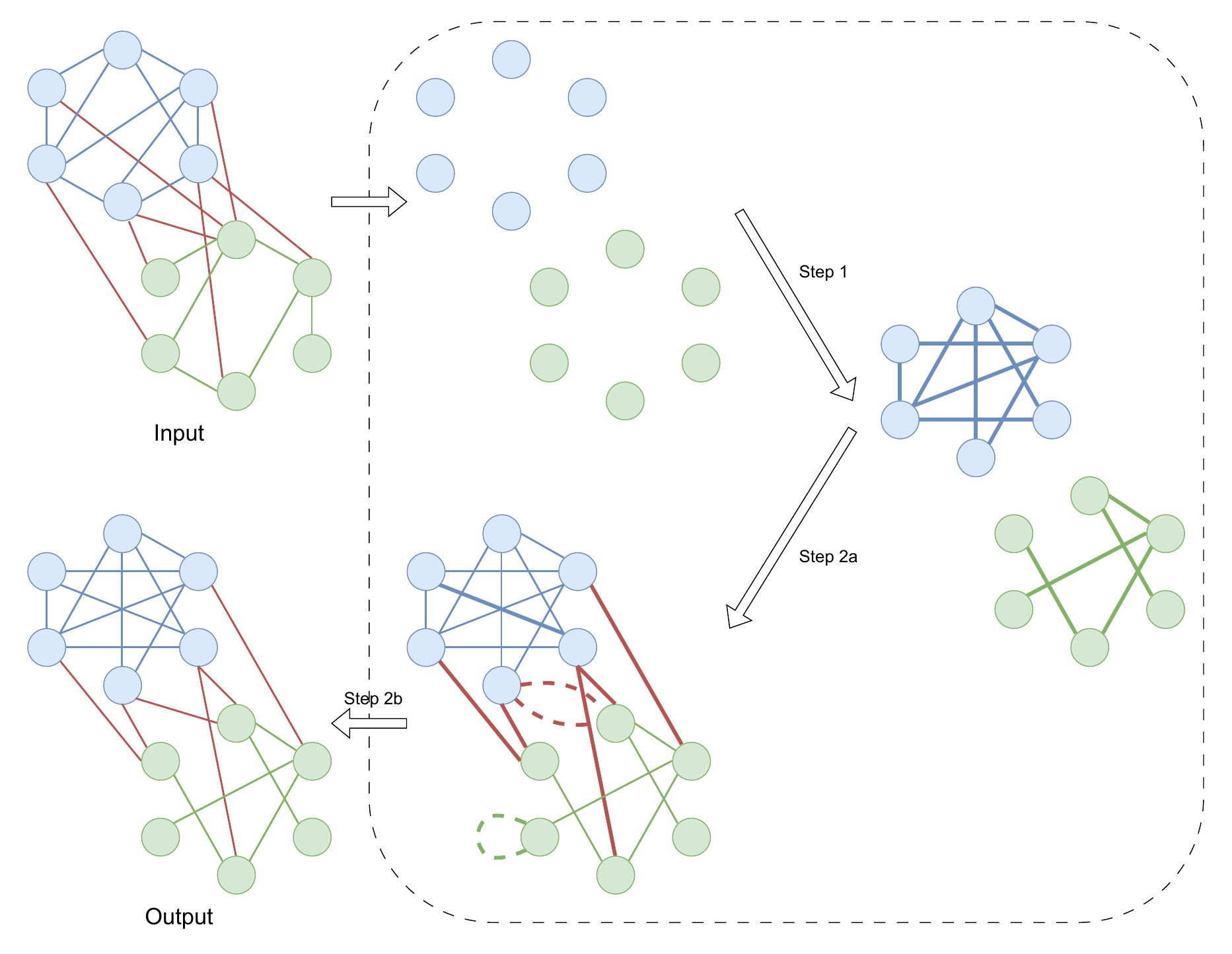}
    \caption{\textbf{Stage 1 of EC-SBM: Generation of the synthetic clustered subnetwork}. The empirical cluster assignment is maintained as the synthetic cluster assignment. In Step 1, we generate for each cluster a $k$-edge-connected subnetwork where $k$ is the desired edge connectivity of that cluster. In Step 2a, we generate the remaining edges according to the updated parameters using SBM; this can result in parallel edges and self-loops (dashed). In Step 2b, we remove the excessive edges to obtain the final output.}
    \label{fig:chambana-stage1}
\end{figure}

For each cluster $C$, to ensure $C$ has edge connectivity of at least $k$, we propose a procedure to generate a spanning subgraph on the set of vertices assigned to the cluster with edge connectivity of at least $k$ (Step 1). The procedure is illustrated in Figure \ref{fig:kec}. It involves initializing a $k$-edge-connected graph (which, to keep it simple, we use the $(k+1)$-clique) and processing the remaining vertices sequentially. Each vertex being processed will be made adjacent to $k$ previously processed vertices. The procedure will concurrently update the other input parameters to maintain the consistency described in Section \ref{sec:prior-works/sbm}, specifically the desired degrees of the vertices assigned to $C$ and the desired edge count matrix in the cell on the diagonal associated with $C$ (the cluster assignment remains the same).

\begin{figure}[!t]
    \centering
    \includegraphics[width=\linewidth]{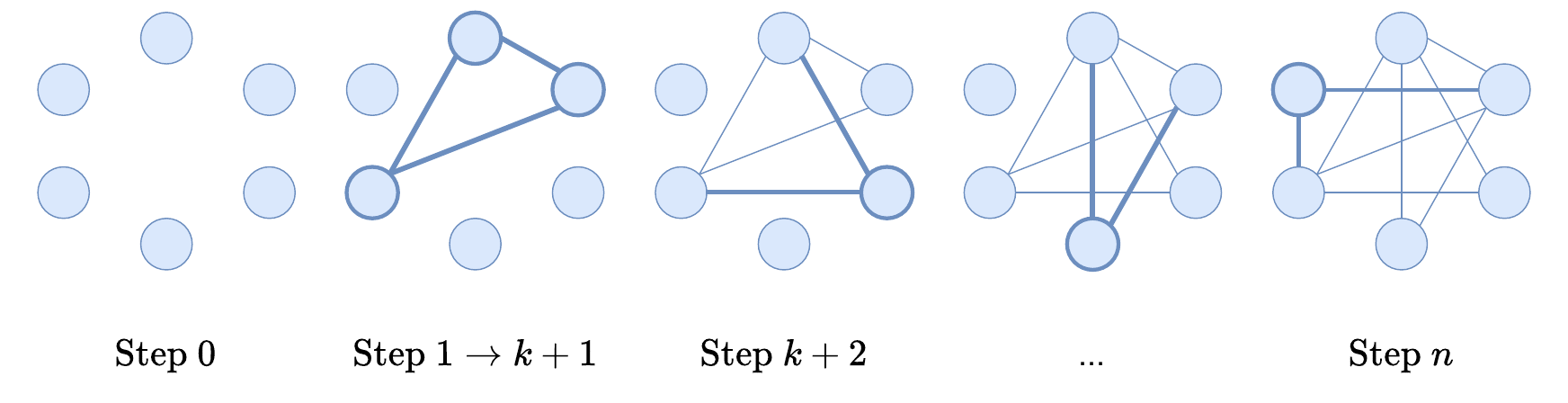}
    \caption{\textbf{Step 1 of Stage 1 of EC-SBM: Generation of a $k$-edge-connected subnetwork.} In the first $k+1$ steps, we construct a $(k+1)$-clique. From Step $k+2$ up to $n$, we process the remaining vertices sequentially by making each vertex adjacent to $k$ previously processed vertices.}
    \label{fig:kec}
\end{figure}

After processing all clusters, we generate the remaining edges using SBM with the updated parameters (Step 2a). As observed in Section \ref{sec:prior-works/sbm}, SBM will generate a multi-graph with parallel edges and self-loops. We removed the self-loops and kept only one edge for each set of parallel edges (Step 2b).

\paragraph{Stage 2: Generation of the synthetic outlier subnetwork}

 We start by considering a subnetwork of the empirical network, which is formed by removing all edges that connect two vertices that are not considered outliers. Additionally, we introduce an auxiliary cluster assignment where each outlier is assigned to its own singleton cluster, alongside the original clusters. Next, we can compute the degree sequence and edge count matrix for this described subnetwork and cluster assignment. These parameters can then be used as input for the SBM to generate a synthetic outlier subnetwork. The resulting synthetic outlier subnetwork, when combined with the synthetic clustered subnetwork from Stage 1, constitutes the partially complete synthetic network.

\paragraph{Stage 3: Degree correction}

To address the mismatch between the generated degree sequence and the desired degree sequence caused by the resolution of parallel edges and self-loops from SBM in Step 2b of Stage 1, we introduce an additional step. This step involves adding edges between vertices that have not yet achieved their desired degree. This step is the same as Stage 4 within the v1 variant of RECCS \cite{reccs}.

\section{Results and Discussion}
\label{sec:results}

\subsection{Experimental study design} 
\label{sec:exp-design}

In this section, we describe a high-level outline of our experimental study design.  

\begin{itemize} 
    \item \textbf{Experiment 1: Finding the best input clusterings for each synthetic network generator: SBM, RECCS, and EC-SBM} For each synthetic network generator, by differing the input clustering, the quality of fit of the synthetic network to the empirical network concerning various network-only statistics (independent of the clustering) changes. This experiment finds for each simulator their best input clusterings for generating the synthetic networks that best fit the corresponding empirical network.
    \item \textbf{Experiment 2: Comparing SBM, RECCS, and EC-SBM on the best input clustering found in Experiment 1} This experiment assesses EC-SBM's performance compared to prior works (SBM and RECCS), each simulator using their best input clusterings.
    \item \textbf{Experiment 3: Runtime analysis of EC-SBM} This experiment finds out how long it takes for EC-SBM to generate a synthetic network, especially when referencing a large network (millions of vertices).
\end{itemize}

The overall procedure is as follows. We take an empirical network and run it through a community detection method to obtain an empirical clustering. Using the empirical network and its respective empirical clustering, we generate a synthetic network and a respective synthetic clustering. We compute the similarity between the empirical and synthetic pair concerning various statistics to evaluate the generator's performance (detailed below). This similarity can be used to compare between different generators. Figure \ref{fig:eval_pipeline} illustrates the entire evaluation process.

\begin{figure}[!t]
    \centering
    \includegraphics[width=\linewidth]{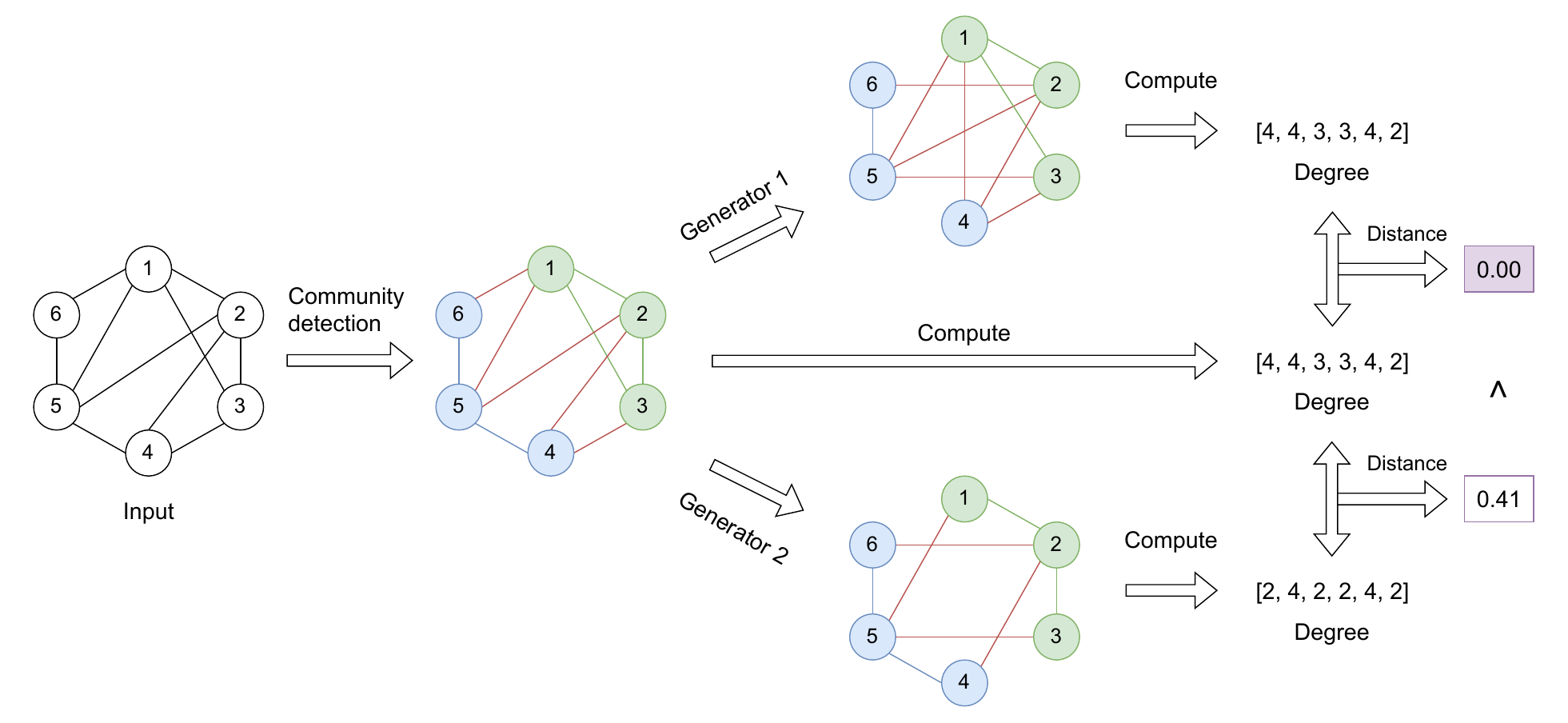}
    \caption{\textbf{The evaluation process of synthetic network generators.} An empirical clustering is obtained using a community detection method on the empirical network. Using the empirical network and clustering, the generator generates a synthetic network and clustering. Various statistics can be computed and compared between the empirical and synthetic pair. The distance computed quantifies the generator's performance and can be used to compare between generators. }
    \label{fig:eval_pipeline}
\end{figure}

The network corpus used for these experiments is described in Section \ref{sec:methods-materials/materials}. For Experiment 1 and Experiment 2, we use the $74$ small and medium networks. For Experiment 3, we use the $3$ large networks.

The community detection methods used in this paper are Leiden-CPM, Leiden-Mod, and SBM (as a community detection method). Specifically with SBM, we used the \texttt{minimize\_blockmodel\_dl} function within \texttt{graph-tool} \cite{graph-tool} as done in \cite{eval_sbm}: we ran three different models (non-degree corrected, degree corrected, and planted partition) and selected the best one (with the smallest description length). We applied post-processing to the clustering results using CM \cite{cm} for Leiden-based methods and CC (Connected Components) or WCC \cite{sbm-wcc} for SBM. These procedures refine clusters by splitting them to meet specific edge connectivity requirements, as outlined in the respective papers.

\begin{table}[!t]
\centering
\caption{Statistics to be computed and compared between the synthetic and the empirical network and cluster assignment. (*) are statistics that do not depend on the cluster assignment.}
\label{tab:statistics}
\begin{tabular}{|l|l|l|}
\hline
\multicolumn{1}{|c|}{\textbf{Statistic}} & \multicolumn{1}{c|}{\textbf{Type and range of value(s)}} & \multicolumn{1}{c|}{\textbf{Description}} \\ \hline
\texttt{pseudo\_diameter} (*) & Scalar $(0, \infty)$ & Approximate diameter \\ \hline
\texttt{char\_time} (*) & Scalar $(0, \infty)$ & Characteristic time of a random walk \\ \hline
\texttt{global\_ccoeff} (*) & Scalar $[0, 1]$ & Global clustering coefficient \\ \hline
\texttt{degree} (*) & Sequence $[0, \infty)$ & Degree of the vertices \\ \hline
\texttt{mixing\_mus} & Sequence $[0, 1]$ & Local mixing parameter of each vertex \\ \hline
\texttt{mincuts} & Sequence $[0, \infty)$ & Edge connectivity of each cluster \\ \hline
\texttt{c\_edge} & Sequence $[0, \infty)$ & Number of edges inside each cluster \\ \hline
\texttt{o\_deg} & Sequence $[0, \infty)$ & Degree of the outliers \\ \hline
\end{tabular}
\end{table}

Table \ref{tab:statistics} lists the statistics of interest. For the positive scalar (\texttt{pseudo\_diameter} and \texttt{char\_time}), the distance used is the signed difference (SD), calculated as $x - x'$, where $x$ represents the value from the empirical network and $x'$ represents the value from the synthetic network. For \texttt{global\_ccoeff}, the distance used is the signed relative difference (SRD), calculated as $(x - x')/x$, where $x$ represents the value from the empirical network and $x'$ represents the value from the synthetic network. Note that in both cases, a negative distance value indicates that the statistic computed on the synthetic pair is higher than that of the empirical pair. For the rest of the statistics, since they are sequences, the distance used is the root mean square error (RMSE), calculated as $\sqrt{\frac{1}{n} \sum_{i=1}^n (x_i - x'_i)^2}$, where $x$ represents the sequence from the empirical network and cluster assignment and $x'$ represents the value from the synthetic network and cluster assignment. The RMSE value is always non-negative.

In Experiment 1, we look at network-only statistics (the four statistics with (*) in Table \ref{tab:statistics}). In \cite{eval_sbm}, the authors conclude that the important properties to look at are the pseudo-diameter, the characteristic time of a random walk, and the clustering coefficient; thus, these are the criteria we will focus on. \cite{eval_sbm} also suggested the number of edges; we refine the idea by looking at the degree sequence of each vertex.  In Experiment 2, we look at all criteria. The additional cluster-specific statistics are those studied in RECCS \cite{reccs}. 

\subsection{Results of Experiment 1}

Figure \ref{fig:comp/val/per_sim/leiden_leidencm_sbm} illustrates the effect of using different input clustering on the similarity between the empirical and synthetic networks generated by SBM, RECCS, and EC-SBM. 
Considering each synthetic network generation method in turn, we see the following trends.

For SBM  (Figure \ref{fig:comp/val/per_sim/leiden_leidencm_sbm} top panel), using SBM+WCC as an input clustering is either best (e.g., pseudo-diameter, characteristic time, and degree)  or close to best (global clustering coefficient) for each criterion.
RECCS shows a different pattern (Figure \ref{fig:comp/val/per_sim/leiden_leidencm_sbm} middle panel), with SBM+WCC and several other input clusterings providing good accuracy for pseudo diameter and global clustering coefficient. For the characteristic time, SBM+WCC is clearly best, but for the degree, SBM+WCC and SBM+CC are worst.
However, Leiden-Mod+CM provides very good accuracy on all criteria compared to the other input clustering methods.
Therefore, for RECCS, we pick Leiden-Mod+CM as a good input clustering.
Finally, for EC-SBM (Figure \ref{fig:comp/val/per_sim/leiden_leidencm_sbm} bottom panel),
SBM+WCC is either best or close to best for the other criteria.
% }
 
Taking these trends into consideration, we select SBM+WCC and Leiden-Mod+CM as two clusterings for use in providing parameters to these three simulators. 

\begin{figure}[!t]
    \includegraphics[width=\linewidth]{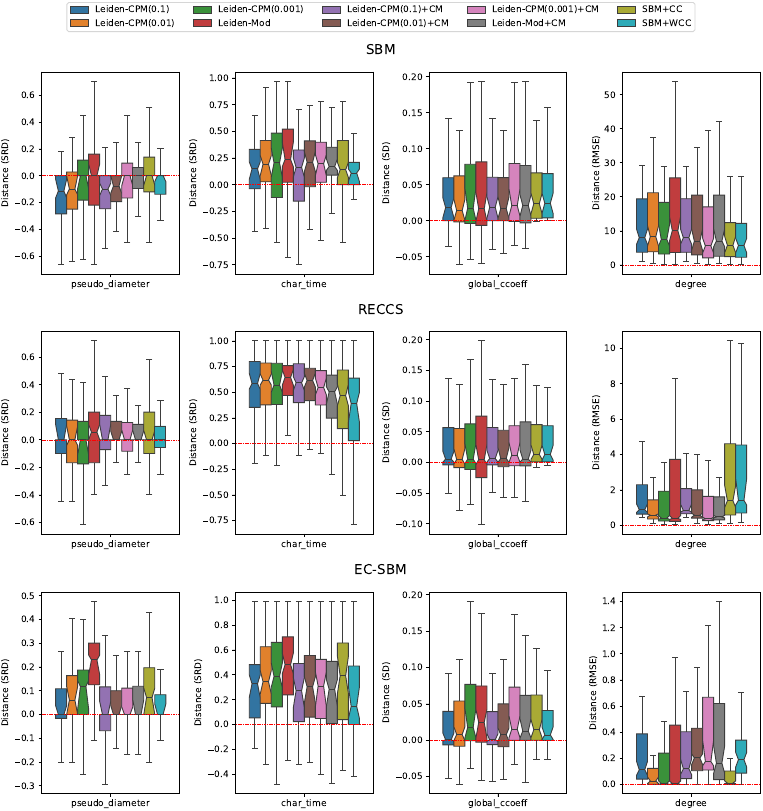}
    \caption{\textbf{Impact of input clustering on the similarity between the clustered input network and synthetic network}. The comparison is
    done on 74 networks with respect to 4 different statistics (see Table \ref{tab:statistics}). A distance value closer to $0.0$ is preferred. For each simulator, there is no single best input clustering across all statistics. For SBM and EC-SBM, SBM+WCC has an advantage when we look at the pseudo-diameter and the characteristic time. For RECCS, SBM+WCC and Leiden-Mod+CM are comparable, with the former performing well on the characteristic time and the latter performing well on the others. 
    }
    \label{fig:comp/val/per_sim/leiden_leidencm_sbm}
\end{figure}

\subsection{Results of Experiment 2}

Figure \ref{fig:comp/val/per_clustering/sbm_reccs_sbmmcs+e/network} and Figure \ref{fig:comp/val/per_clustering/sbm_reccs_sbmmcs+e/cluster} compares the performance of SBM, RECCS, and EC-SBM using the two best input clusterings, SBM+WCC and Leiden-Mod+CM, across $4$ network-only criteria (Figure \ref{fig:comp/val/per_clustering/sbm_reccs_sbmmcs+e/network}) and $4$ cluster-specific criteria (Figure \ref{fig:comp/val/per_clustering/sbm_reccs_sbmmcs+e/cluster}).

SBM using either input clustering performs poorly on most criteria, except for being the best at the characteristic time and comparable to RECCS at the outlier degree sequence. On the other hand, EC-SBM using SBM+WCC input clustering is a well-rounded pipeline with a strong performance in most statistics. Most notably, it outperforms RECCS in terms of both the degree sequence and the outlier degree sequence.

However, EC-SBM using SBM+WCC input clustering is not the best. Regarding the characteristic time, EC-SBM using SBM+WCC input clustering is only the second best after SBM using SBM+WCC. Regarding the global coefficient, RECCS using Leiden-Mod+CM has the lowest median; however, EC-SBM using SBM+WCC input clustering is only marginally worse and has less variance. Regarding edge connectivity of the clusters, for either the SBM+WCC or Leiden-Mod+CM input clustering, RECCS is the best, followed by EC-SBM.

\begin{figure}[t]
    \centering
    \includegraphics[width=\linewidth]{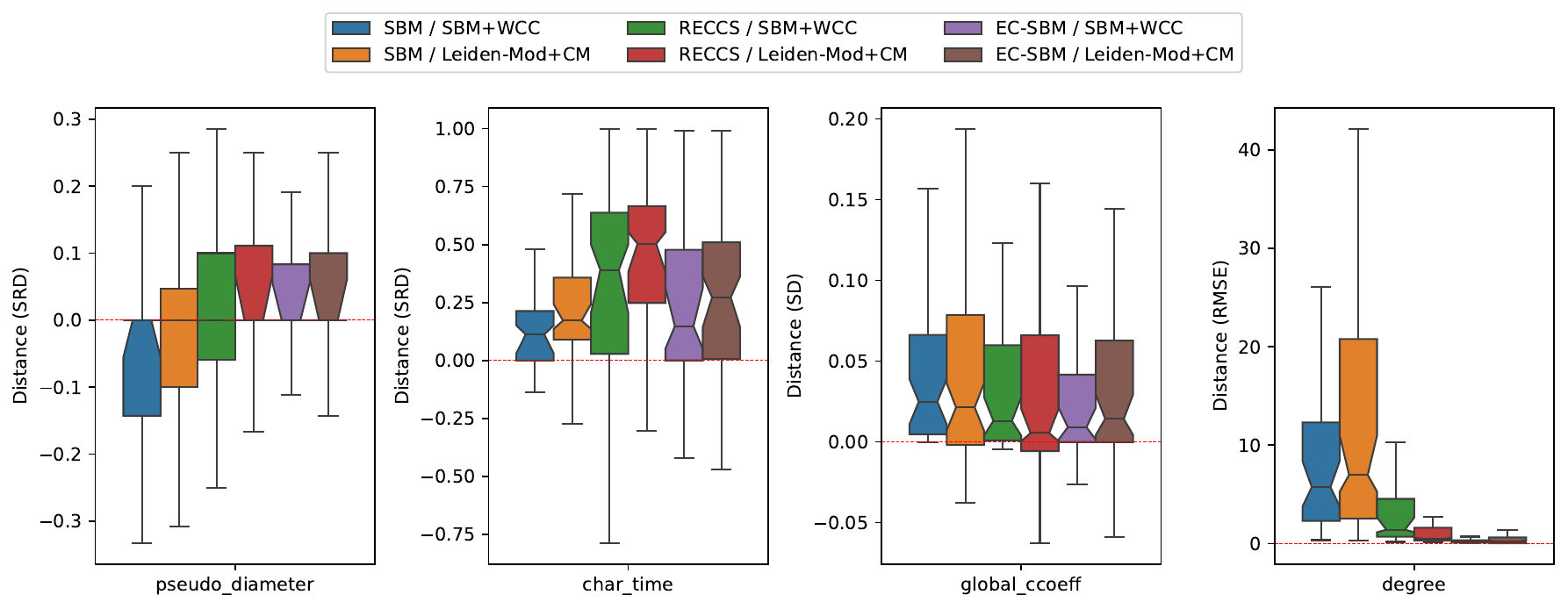}
    \caption{\textbf{Comparison between three simulators using two input clusterings on network-only criteria.} The simulators are SBM, RECCS, and EC-SBM. The input clusterings are SBM+WCC and Leiden-Mod+CM, which we determine to be the most suitable in Experiment 1. The comparison is done on $74$ networks with respect to $4$ network-only criteria. 
    }
    \label{fig:comp/val/per_clustering/sbm_reccs_sbmmcs+e/network}
\end{figure}

\begin{figure}[t]
    \centering
    \includegraphics[width=\linewidth]{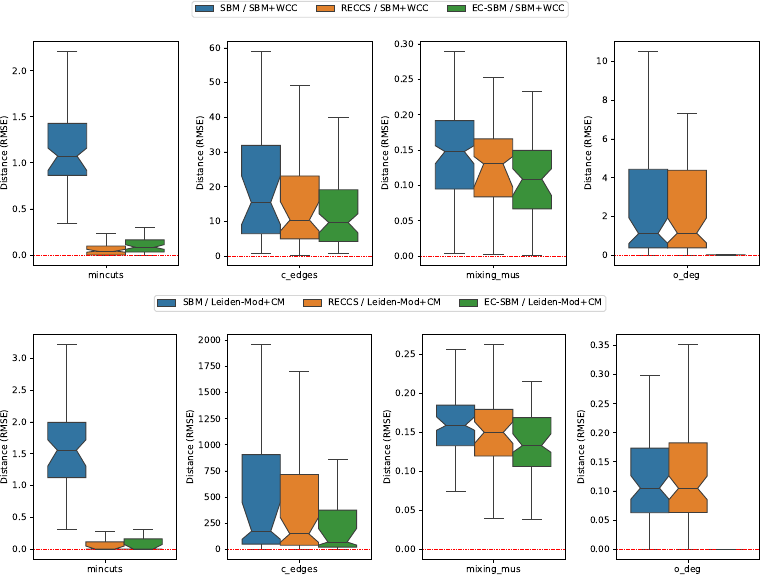}
    \caption{\textbf{Comparison of three simulators and two input clusterings on cluster-specific criteria.} The simulators are SBM, RECCS, and EC-SBM. The input clusterings are SBM+WCC (top) and Leiden-Mod+CM (bottom), which we determine to be the most suitable in Experiment 1. The comparison is done on $74$ networks with respect to $4$ cluster-specific criteria. We group the pipelines according to their input clustering. The distance values for \texttt{o\_deg} of EC-SBM using SBM+WCC and Leiden-Mod+CM input clustering are very close to $0.0$ in this figure.  
    }
    \label{fig:comp/val/per_clustering/sbm_reccs_sbmmcs+e/cluster}
\end{figure}

\clearpage

\subsection{Results of Experiment 3}

Table \ref{tab:time/large} shows the runtime analysis of the generation process of EC-SBM and RECCS using these networks as reference networks and SBM+WCC input clustering. For each simulator, we measure the time from getting the empirical network and clustering to generating a synthetic network and clustering. For SBM, this includes the generation described in Section \ref{sec:prior-works/sbm}. For RECCS, this includes generating the initial synthetic network using SBM and the post-processing steps as described in \cite{reccs}. For EC-SBM, this includes the three-stage generation described in Section \ref{sec:design-ecsbm}.

\begin{table}[!ht]
\centering
\caption{\textbf{Runtime analysis of the generation process of SBM, RECCS, and EC-SBM.}}
\label{tab:time/large}
\begin{tabular}{|l|r|r|r|}
\hline
 & \multicolumn{1}{c|}{\texttt{CEN}} & \multicolumn{1}{c|}{\texttt{orkut}} & \multicolumn{1}{c|}{\texttt{livejournal}} \\ \hline
\textbf{SBM (hours)} & $\approx$ 1.17 & $\approx$ 1.44 & $\approx$ 0.52 \\ \hline
\textbf{RECCS (hours)} & $\approx$ 1.62 & $\approx$ 3.50 & $\approx$ 2.55 \\ \hline
\textbf{EC-SBM (hours)} & $\approx$ 5.00 & $\approx$ 3.77 & $\approx$ 1.47 \\ \hline
\hline
 % & \multicolumn{3}{c|}{Statistics} \\ \hline
\textbf{Number of nodes} & 14.0M & 3.1M & 4.9M \\ \hline
\textbf{Number of outliers} & 12.8M & 0.6M & 2.0M \\ \hline
\textbf{Number of edges} & 92.1M & 117.2M & 42.9M \\ \hline
\textbf{Number of clusters} & 12K & 31K & 134K \\ \hline
\end{tabular}
The comparison is done using large networks and their corresponding SBM+WCC input clustering. Here, the generation process starts with being given the empirical network and clustering and ends with producing the synthetic network and clustering. The statistics of the network and clustering are also included.
\end{table}

Overall, the generation process is fast, even on large networks. The unmodified SBM is the fastest.  Between RECCS and EC-SBM, the absolute and relative runtime depends on the input network and clustering. RECCS is faster for \texttt{CEN} but slower for \texttt{livejournal}; they are similar for \texttt{orkut}, with RECCS being faster.

For reference, while it took roughly $5$ hours to generate a synthetic network based on \texttt{CEN} and SBM+WCC input clustering, it took $158$ hours to compute the SBM clustering for \texttt{CEN} ($38$ hours for the degree-corrected model, $66$ hours for the non-degree-corrected model, and $54$ hours for the planted partition models) and an additional $1.5$ hours to compute the WCC processing on top of it \cite{sbm-wcc}.

\subsection{Summary}

We find SBM+WCC and Leiden-Mod+CM to be the two best clustering methods to obtain the input clustering for SBM, RECCS, and EC-SBM for simulating network-only statistics. 
For these input clusterings, of the three network simulators we study, EC-SBM using SBM+WCC is overall the best at coming close to the clustered real-world network-only and cluster-specific statistics. Finally, we find that EC-SBM can generate large networks (with millions of nodes) using the SBM+WCC input clustering.

Although SBM+WCC provides overall the closest results, depending on the user's needs, one of the other simulators may be better suited. 
If a faster simulator is required, SBM may be preferable despite its lower accuracy in reproducing network statistics and clustering statistics. 
RECCS and EC-SBM each have aspects in which one is better than the other, so the choice between them can depend on what is most important to the user.
For example, RECCS is slightly better than EC-SBM at providing a close fit for the edge connectivity of the clusters, while EC-SBM is better for the mixing parameter, the characteristic time of a random walk, and the degree sequence. 
Finally, the runtimes of RECCS and EC-SBM are close, with an advantage to RECCS. 
Thus, the choice between RECCS and EC-SBM depends on the relative importance of runtime and edge-connectivity compared to the other network and clustering statistics. Finally, using \texttt{graph-tool} for SBM construction is the fastest of these synthetic network generation methods; hence, if runtime is of paramount importance, then \texttt{graph-tool} should be preferred.

\section{Conclusions}
\label{sec:conclusions}

Synthetic networks with ground truth clusterings are needed for multiple purposes, including the evaluation of community detection and community search methods. 
Stochastic Block Models (SBMs) are established models for this problem, and prior work \cite{eval_sbm} has shown that SBMs generated using \texttt{graph-tool} \cite{graph-tool} are able to reproduce with high accuracy many network features of real-world networks.
Both \cite{reccs} and our study show that the ground truth clusters produced using SBMs are often internally disconnected, which is undesirable.  

Motivated by the value of having synthetic networks with quality ground truth clusters, we presented EC-SBM, a new simulator tool that is designed to improve upon the utility of software for SBMs for generating synthetic networks.
We proposed EC-SBM, a simulator that takes as input parameters about a clustered network and generates a synthetic network that aims to approximate the features of that input clustered network. We designed EC-SBM by focusing on simulating the edge connectivity of the clusters. Additionally, we incorporated a post-processing step to align the degree sequence of the synthetic network with that of the empirical network. 

Our empirical analyses show that both EC-SBM and RECCS are better than SBMs for the cluster edge connectivity and the degree sequence while coming close in the other criteria.  We also show that overall, EC-SBM is superior to RECCS: more accurate for several parameters, such as the characteristic time and degree sequence, and only slightly less accurate than RECCS for cluster edge connectivity.

As EC-SBM is built on the implementation of SBM in \texttt{graph-tool} \cite{graph-tool}, we have not explored other SBM implementations or alternative synthetic network generators such as nPSO \cite{npso} or LFR \cite{lfr}/ABCD \cite{abcd} in this paper. In future work, we will explore analyzing and potentially extending the idea of EC-SBM to other generators to compare their performances.

Furthermore, as this paper focuses mainly on generating realistic synthetic networks, we have not studied the performance of clustering methods on the generated synthetic networks;  this is an important direction for future work.
While we have verified that EC-SBM can construct synthetic networks with millions of nodes, we have not evaluated the limits to this scalability. 
Given the increasing interest in understanding community structure in very large networks, with potentially billions of nodes, this is also an important direction for future work.

\section{Materials and Methods}
\label{sec:methods-materials}

\subsection{Materials}
\label{sec:methods-materials/materials}

The corpus of networks studied in this paper contains $74$ small and medium-sized networks (from approximately $1,000$ to approximately $1.4$ million vertices) and $3$ large networks (more than $3$ million vertices). Most of the networks are sourced from \cite{peixoto-db}, except for the Curated Exosome Network \texttt{CEN} and \texttt{orkut}, which are used in \cite{cm} and available from the publication. 

We restrict the study to non-bipartite networks because the community detection methods involved are not designed to work on bipartite graphs.

The $74$ small and medium-sized networks: \texttt{academia\_edu}, \texttt{advogato}, \texttt{anybeat}, \texttt{at\_migrations}, \texttt{berkstan\_web}, \texttt{bible\_nouns}, \texttt{bitcoin\_alpha}, \texttt{bitcoin\_trust}, \texttt{chess}, \texttt{chicago\_road}, \texttt{citeseer}, \texttt{collins\_yeast}, \texttt{cora}, \texttt{dblp\_cite}, \texttt{dblp\_coauthor\_snap}, \texttt{dnc}, \texttt{douban}, \texttt{drosophila\_flybi}, \texttt{elec}, \texttt{email\_enron}, \texttt{email\_eu}, \texttt{epinions\_trust}, \texttt{faa\_routes}, \texttt{facebook\_wall}, \texttt{fediverse}, \texttt{fly\_hemibrain}, \texttt{fly\_larva}, \texttt{foldoc}, \texttt{google}, \texttt{google\_plus}, \texttt{google\_web}, \texttt{hyves}, \texttt{inploid}, \texttt{interactome\_figeys}, \texttt{interactome\_stelzl}, \texttt{interactome\_vidal}, \texttt{internet\_as}, \texttt{jdk}, \texttt{jung}, \texttt{lastfm\_aminer}, \texttt{libimseti}, \texttt{linux}, \texttt{livemocha}, \texttt{lkml\_reply}, \texttt{marker\_cafe}, \texttt{marvel\_universe}, \texttt{myspace\_aminer}, \texttt{netscience}, \texttt{new\_zealand\_collab}, \texttt{notre\_dame\_web}, \texttt{openflights}, \texttt{petster}, \texttt{pgp\_strong}, \texttt{polblogs}, \texttt{power}, \texttt{prosper}, \texttt{python\_dependency}, \texttt{reactome}, \texttt{slashdot\_threads}, \texttt{slashdot\_zoo}, \texttt{sp\_infectious}, \texttt{stanford\_web}, \texttt{topology}, \texttt{twitter}, \texttt{twitter\_15m}, \texttt{uni\_email}, \texttt{us\_air\_traffic}, \texttt{wiki\_link\_dyn}, \texttt{wiki\_rfa}, \texttt{wiki\_users}, \texttt{wikiconflict}, \texttt{word\_assoc}, \texttt{wordnet}, \texttt{yahoo\_ads}. 

The $3$ large networks: \texttt{CEN}, \texttt{livejournal}, \texttt{orkut}.

\subsection{The EC-SBM method}
\label{sec:methods-materials/method}

In this section, we describe our proposed simulator in more detail and more formally. The three-stage process is described in Section \ref{sec:methods-materials/method/three-stages} with the first two stages being described in Section \ref{sec:methods-materials/method/stage1} and \ref{sec:methods-materials/method/stage2}. The third stage is identical to Stage IV of Step 1 of RECCS version 1 \cite{reccs}.

\subsubsection{A three-stage generation of the synthetic network}
\label{sec:methods-materials/method/three-stages}

We are given an empirical network $G$ with $n$ vertices and an empirical cluster assignment $\mathcal{C}$ with $m$ clusters. Some vertices may not belong to any cluster in $\mathcal{C}$; these are considered outliers. We define $G_c$ as the subnetwork induced by the vertices that are not considered outliers and we refer to this as the clustered subnetwork. On the other hand, let $G_o$ be the subnetwork formed by removing all edges of $G_c$ from the empirical network $G$ and call this the outlier subnetwork. 

In the first stage, we generate a synthetic version of the clustered subnetwork, denoted by $G_c^\prime$, and a synthetic clustering assignment $\mathcal{C}^\prime$ which simulates $G_c$ and $\mathcal{C}$. In the second stage, we generate a synthetic outlier subnetwork $G_o^\prime$ to simulate $G_o$. The partially complete synthetic network is formed by combining these two: $\tilde{G}^\prime = G_c^\prime \cup G_o^\prime$. In the third stage, we add edges to $\tilde{G}^\prime$ so that the degree of each vertex matches that of the corresponding vertex in $G$. This process yields our final synthetic network, $G^\prime$.

\subsubsection{Generation of the clustered subnetwork}
\label{sec:methods-materials/method/stage1}

In this section, we describe our proposed method for generating a synthetic clustered subnetwork.

Firstly, we keep the same cluster assignment in the synthetic network as in the empirical network, meaning the synthetic cluster assignment $\mathcal{C}^\prime$ is equal to $\mathcal{C}$.

Secondly, we compute the required degree sequence $d$ (a vector of size $n$) and the required edge count matrix $e$ (a matrix of size $m \times m$) from $G_c$. Additionally, we compute the edge connectivity sequence $\lambda$, which is a vector of size $m$ representing the edge connectivity of the clusters in the empirical pair $(G_c, \mathcal{C})$. This sequence indicates the desired edge connectivity for each cluster. Specifically, we aim to generate a synthetic clustered subnetwork $G^\prime_c$ such that $\mathcal{C}^\prime$ accompanies $G_c^\prime$, and for all clusters $C^\prime_i \in \mathcal{C}^\prime$, the subnetwork $G^\prime_c(C^\prime_i)$ is $\lambda_i$-edge-connected.

We propose a procedure to generate a spanning subnetwork from a set of vertices, ensuring that each cluster has an edge connectivity of at least a specified value.

\paragraph{Generation of a $k$-edge-connected spanning subnetwork} 

Given a set of vertices $V$ of size $n$ and an integer $k < n$, we can generate a $k$-edge-connected spanning subnetwork on $V$ via the following procedure, called \texttt{GEN-KECSSN}$(V, k)$.
\begin{enumerate}
    \item We arbitrarily label the vertices $v_1, \dots, v_n$.
    \item We initialize with a $(k+1)$-clique, denoted by $G_{k+1} = (V_{k+1}, E_{k + 1})$. Specifically, we sequentially add vertices $v_1, \dots, v_{k+1}$ to an empty graph, each vertex being added is connected to all previously added vertices. In other words, $V_{k+1} = \{ v_1, \dots v_{k + 1} \}$ and $E_{k+1} = \{ (v_i, v_j): 1 \leq i < j \leq k + 1 \}$.
    \item At each step, we add a new vertex to the graph and connect it with $k$ previously added vertices. More formally, for $i = k+2, \dots, n$, let $G_i = (V_i, E_i)$ where $V_i = V_{i-1} \cup \{ v_{i} \}$ and $E_i = E_{i-1} \cup \{ (v_i, v): v \in S_i \}$ where $S_i$ is a set of $k$ arbitrarily chosen vertices from $V_{i-1}$ (thus, $S_i \subset V_{i-1}, |S| = k$).
    \item Output $G_n$.
\end{enumerate}

The described procedure has several flexible components that can be explored for improvement. For initialization, any $k$-edge-connected graph can be used instead of a $(k+1)$-clique. We specifically chose the $(k + 1)$-clique as a straightforward starting $k$-edge-connected graph. Additionally, the vertices can be ordered in various ways; we specifically choose to process them in decreasing order of their degree. The method for selecting neighbors while processing a vertex also allows for flexibility. In our implementation, we randomly select the vertices with probabilities proportional to their availability, where the availability of a vertex at a given step is determined by how many more edges it can form to reach its desired degree.

With the following result, we show that this procedure will give a $k$-edge-connected graph as desired.

\begin{theorem}
    Let $\lambda(G_n)$ be the minimum edge-cut size of $G_n$. Then, $\lambda(G_n) \geq k$.
\end{theorem}

\begin{proof}
Let $(S, T)$ be the minimum edge-cut of $G_n = (V_n, E_n)$. Recall that $v_1, \dots, v_n$ is the order in which the vertices are added and let $V_i = \{v_1, \dots, v_i\}$ for $1 \leq i \leq n$. 

Recall that the first $k + 1$ vertices are made into a $(k + 1)$-clique denoted $G_{k+1}$. Let $d = |\{ 1 \leq j \leq k + 1: v_j \in T \}|$, i.e., the number of vertices from the initial clique belonging to $T$. Without loss of generality, let $T$ have the smaller number of vertices in the $(k+1)$-clique. Then, $d \leq k + 1$ and $|V_{k+1} \cap S| = k - d + 1$.

\underline{Case 1}: $d > 0$. WLOG, let $v_1, \dots, v_{k-d+1} \in S$ and $v_{k-d+2}, \dots, v_{k+1} \in T$.

Then $E$ contains $\{ (v_i, v_j): 1 \leq i \leq k - d + 1, k - d + 2 \leq j \leq  k + 1\} =: C$. Consequently, the cutset of $(S, T)$ contains $C$ since $v_1, \dots, v_{k-d+1} \in S$ and $v_{k-d+2}, \dots, v_{k+1} \in T$. $|C| = d(k - d + 1) \geq k$ because $d(k-d+1) - k = (k - d)(d - 1) \geq 0$.

Thus, $\lambda(G) \geq |C| \geq k$.

\underline{Case 2}: $d = 0$. Then $V_{k + 1} \subset S$.

Let $t = \min \{ s: v_s \in T \}$, i.e, $v_t$ is the earliest vertex added to $T$. Let $U_t = \{ u: (u, v_t) \in E, u \in V_{t-1} \} \subset S$, i.e, $U_t$ is the set of $k$ vertices that are made adjacent to $v_t$ when $v_t$ is being added. 

Then $E$ contains $\{ (u, v_t) : u \in U \} =: C'$. Consequently, the cutset of $(S, T)$ contains $C'$ since $u \in U \subset S$ and $v_t \in T$. 

Thus, $\lambda(G) \geq |C'| = |U| = k$.
\end{proof}

In Step 1, we process each cluster $C_i^\prime \in \mathcal{C}^\prime$ sequentially and independently. Let $G^{\prime \text{mcs}}_{C_i^\prime} = \text{\texttt{GEN-KECSSN}}(C_i^\prime, \lambda_i)$, i.e., we generate a spanning subnetwork on $C_i^\prime$ that is guaranteed to be $\lambda_i$-edge-connected.

As we generate the spanning subnetwork for a cluster $C_l, l \in \{1, \dots, m\}$, we must correspondingly decrease the required degree of each vertex of $C_i$ (i.e., $d_i$ for each $i \in C_l$) and the required edge count inside $C$ (i.e., $e_{l,l}$). In other words, when we create an edge between vertex $i$ and vertex $j$ in for the spanning subnetwork of cluster $C_l$, we decrease $d_i$ and $d_j$ by one each and decrease $e_{l, l}$ by two. There are cases where these decreases lead to a negative value for either the degree or the edge count. In such cases, we reverse the decrease by increasing $d_i$ and $d_j$ by one each and increasing $e_{l, l}$ by two.

After processing all clusters, we take the union of all the synthetic subnetworks to get a synthetic spanning subnetwork, i.e., $G_c^{\prime \text{mcs}} = \bigcup_{C^\prime_i \in \mathcal{C}} G^{\prime \text{mcs}}_{C_i^\prime}$.

 In Step 2a, using SBM with the updated input parameters, we generate a multi-graph $\tilde{G}_c^{\prime \text{base}}$. In Step 2b, we simplify $\tilde{G}_c^{\prime \text{base}}$ by removing all the self-loops and removing all excessive parallel edges (leaving one edge for each set of parallel edges between a pair of vertices) to obtain the synthetic subnetwork $G_c^{\prime \text{base}}$.
 
 We take $G^\prime_c = G_c^{\prime \text{mcs}} \cup G_c^{\prime \text{base}}$ as our synthetic clustered subnetwork. Note that since it is a union, the clusters of $G_c^\prime$ will maintain their connectivity irrespective of $G_c^{\prime \text{base}}$.

\subsubsection{Generation of the outlier subnetwork}
\label{sec:methods-materials/method/stage2}

In this section, we describe the process of generating the synthetic outlier subnetwork.

Let $V_o$ be the set of outliers and $E_o$ be the set of edges in $G_o$. Hence, $E_o$ only contains the edges where at least one of the ends is an outlier. Let $\bar{\mathcal{C}} = \mathcal{C} \cup \left( \cup_{v \in V_o} \{ \{ v \} \} \right)$ be the auxiliary cluster assignment for the outlier subnetwork. Then, $\bar{\mathcal{C}}$ accompanies $G$ by treating each outlier as a singleton cluster. 

From $G_o$ and $\bar{\mathcal{C}}$, we can compute the input parameters to SBM. We use SBM to generate a synthetic multi-graph and resolve the parallel edges and self-loops. We take the resulting network as our synthetic outlier subnetwork $G^\prime_o$.

\section{Declarations}
\subsection{Availability of data and materials} 

The real-world networks on which the analyses are based are already in the public domain. 
The EC-SBM software used is in the public domain at \url{https://github.com/illinois-or-research-analytics/ec-sbm}. The commands used to generate the data are in this paper.

\subsection{Competing interests}
The authors declare that they have no competing interests.

\subsection{Funding}
This work was funded in part by the Insper-Illinois collaboration grant to TW and GC.  

\subsection{Authors' contributions}
TW proposed the basic approach, supervised the research, evaluated experimental results, and edited the manuscript.
TVL designed the algorithm, developed the open-source software, established theoretical properties, performed the simulation study evaluating EC-SBM in comparison to other methods, evaluated results, and wrote the first draft of the paper.
LA provided RECCS synthetic networks.
GC proposed the problem, supervised the research, evaluated experimental results, and edited the manuscript.

\subsection{Acknowledgments}
The authors thank the members of the Warnow-Chacko lab for helpful feedback.

\bibliography{ref}

\end{document}